\documentclass[conference,twocolumn]{IEEEtran}

\newcommand\R{\mathrm{R}}

\usepackage{dsfont}
\usepackage{moreverb}
\usepackage{epsfig}
\usepackage{amsmath,amssymb,amsthm,mathrsfs,amsfonts,dsfont}
\usepackage{adjustbox,lipsum}
\usepackage{algorithm}
\usepackage{algpseudocode}
\usepackage{amsfonts}
\usepackage{amssymb}
\usepackage{amsmath}
\usepackage{amsthm}
\usepackage{subfigure}
\usepackage{multirow}
\usepackage{rotating}
\usepackage{graphicx}
\usepackage{tabularx}
\usepackage{array}
\usepackage{anyfontsize}
\usepackage{color,soul}
\usepackage{graphicx,dblfloatfix}
\usepackage{epstopdf}
\usepackage{blindtext}
\usepackage{amsmath}
\usepackage{amsthm,amssymb,amsmath,bm}
\usepackage{amsfonts}
\usepackage{epsfig}
\usepackage{amssymb}
\usepackage{amsmath}
\usepackage{cite}
\usepackage{graphicx}
\usepackage{fancyhdr}
\usepackage[font={small}]{caption}
\usepackage{tabularx}
\usepackage{tcolorbox}
\usepackage{cite}
\usepackage{soul}
\usepackage{setspace}
\usepackage{enumitem}

\DeclareMathOperator*{\argmin}{arg\,min}
\newtheorem{theorem}{Theorem}
\newtheorem{lemm}{Lemma}

\newtheorem{Pro}{Proposition}

\def\blue{\textcolor{blue}}
\def\red{\textcolor{red}}

\def\yellow{\textcolor{yellow}}

\def\tran{^{\mbox{\scriptsize T}}}

\usepackage[left=0.624in ,right=0.668in ,top=0.72in,bottom=1.016in]{geometry}

\usepackage[hidelinks]{hyperref}

\allowdisplaybreaks

\graphicspath{{Figure/}}

\usepackage{pifont}
\newcommand\blfootnote[1]{%
  \begingroup
  \renewcommand\thefootnote{}\footnote{#1}%
  \addtocounter{footnote}{-1}%
  \endgroup
}

\begin{document}

\sloppy

 
 




\title{On the Age-Optimality of Relax-then-Truncate Approach under Partial Battery Knowledge  in Energy Harvesting IoT Networks}





\author{
\IEEEauthorblockN{Mohammad Hatami\IEEEauthorrefmark{1}, Markus Leinonen\IEEEauthorrefmark{1}, and Marian Codreanu\IEEEauthorrefmark{2}}
}

\maketitle

\begin{abstract}
We consider an
energy harvesting (EH) IoT network, where users make on-demand requests to a cache-enabled edge node to send status updates about various random processes, each monitored by an EH sensor. 
The edge node serves users' requests by either commanding the corresponding sensor to send a fresh status update or retrieving the most recently received measurement from the cache.
We aim to find a control policy at the edge node that minimizes the average on-demand AoI over all sensors
subject to {per-slot transmission and energy constraints} under {partial battery knowledge} at the edge node.
Namely, the limited radio resources (e.g., bandwidth) causes that only a limited number of sensors can send status updates at each time slot (i.e., per-slot transmission constraint) and the scarcity of energy for the EH sensors imposes an energy constraint.
{Besides, the edge node is informed of the sensors' battery levels only via received status update packets, leading to uncertainty about the battery levels for the decision-making.}
We develop a low-complexity algorithm -- termed {relax-then-truncate} -- and prove that it is {asymptotically optimal} as the number of sensors goes to infinity.
Numerical results illustrate that the proposed method achieves significant gains over a request-aware greedy policy and show that it has near-optimal performance even for moderate numbers of sensors. 
\end{abstract}



\sloppy




\section{Introduction}
\blfootnote{\IEEEauthorrefmark{1}Centre for Wireless Communications, University of Oulu, Finland. \\\IEEEauthorrefmark{2}Department of Science and Technology, Link\"{o}ping University, Sweden.\\
{{
This research has been financially supported by the Infotech Oulu, the Academy of Finland (grant 323698), and Academy of Finland 6G Flagship program (grant 346208). The work of M. Leinonen has also been financially supported in part by the Academy of Finland (grant 340171). M. Hatami would like to acknowledge the support of Nokia Foundation.}
}
}
{Internet of Things (IoT) is a key technology to connect different devices to enable emergent applications (e.g., smart society \cite{Xu2014IoTSurvey}) with minimal human intervention.}
In IoT sensing networks, sensors measure physical quantities (e.g., speed) and send measurements to a destination for further processing.
To counteract sensors' severe energy limitations, \textit{energy harvesting} (EH), relying on, e.g., solar or RF ambient sources, is often employed. 
Moreover, reliable control actions in time-critical IoT applications (e.g., drone control and industrial monitoring) require high \textit{freshness} of information at the destination, often quantified by the \textit{Age of Information} (AoI) \cite{AoI_Orginal_12}.
To summarize, these emerging applications require designing \textit{AoI-aware status updating control} that both guarantees timely status delivery and accounts for the limited energy resources of EH sensors.

AoI-aware scheduling has been under intensive research over the last few years.
The works \cite{Hsu_modiano2020aoimultiuser_tcm,tang2020CMDPTRUNCATE,Ceran2021MultiUserCMDP,yao2020age_ISIT,gong2020AoI-Random-Arrival,shao2020partially,stamatakis2022semantics} consider a sufficient power source 
whereby an update can be sent any time.
Differently,
\cite{Stamatakis2019control,ceran2021learningEH,leng2019AoIcognitive,Elvina2021SourceDiversityEH,abd2019reinforcement}
consider that the {source nodes} are powered by \textit{energy harvested from the environment}; thus, AoI-aware scheduling is carried out under the energy causality constraint at the source nodes.
Also, while the above works (implicitly) assume that time-sensitive information
is needed at the destination \textit{at all time moments}, \cite{hatami2020aoi,hatami2021spawc,Hatami2022AsymptoticallyOpt,hatami2022JointTcom,hatami2022spawc_partialbattery,chiariotti2021query,Li2021waiting} 
study the information freshness of the source(s) driven by \textit{users' requests}.

We consider an IoT network that consists of multiple EH sensors, a cache-enabled edge node, and multiple users. Users are interested in timely information about physical quantities (e.g., speed or temperature), each measured by a sensor. Users send requests to the edge node which maintains the most recently received measurements from each sensor. To serve a user's request, the edge node either commands the sensor to send a fresh status update or uses the aged measurement from the cache.
{This imposes \textit{a trade-off} between the information freshness and the energy status of the sensors' batteries, creating challenges to the design of an AoI-aware status update control policy herein.}
Due to the limited amount of radio resources (e.g., bandwidth) in an IoT network, we consider that only a portion 
of sensors can send status updates 
at each slot, imposing a \textit{per-slot transmission constraint}.
Furthermore, in contrast to the prior works (e.g., \cite{hatami2020aoi,hatami2022JointTcom,abd2019reinforcement}), we consider a practical scenario where the edge node is informed of the sensors' battery levels only via received status updates, giving rise to \textit{partial battery knowledge} at the edge node.

We aim to find an \textit{optimal policy} (the best action of the edge node at each time slot) that minimizes the average on-demand AoI over all sensors subject to the per-slot transmission and energy constraints under partial battery knowledge at the edge node.
We propose an asymptotically optimal low-complexity algorithm -- termed \textit{relax-then-truncate} -- and show that it performs close to the optimal solution.
As the main novelty, this paper extends the use of the relax-then-truncate approach introduced in \cite{hatami2022JointTcom} to the scenario where the decision-making relies only on partial battery knowledge. 
This introduces substantial additional challenges to the optimization as it combines the notions of a constrained Markov decision process (CMDP) and a partially observable MDP (POMDP).

\section{System Model and Problem Formulation}\label{sec_systemmodel}




\subsection{Network Model}\label{sec_network}
\begin{figure}[t!]
\centering
\includegraphics[width=.98\columnwidth]{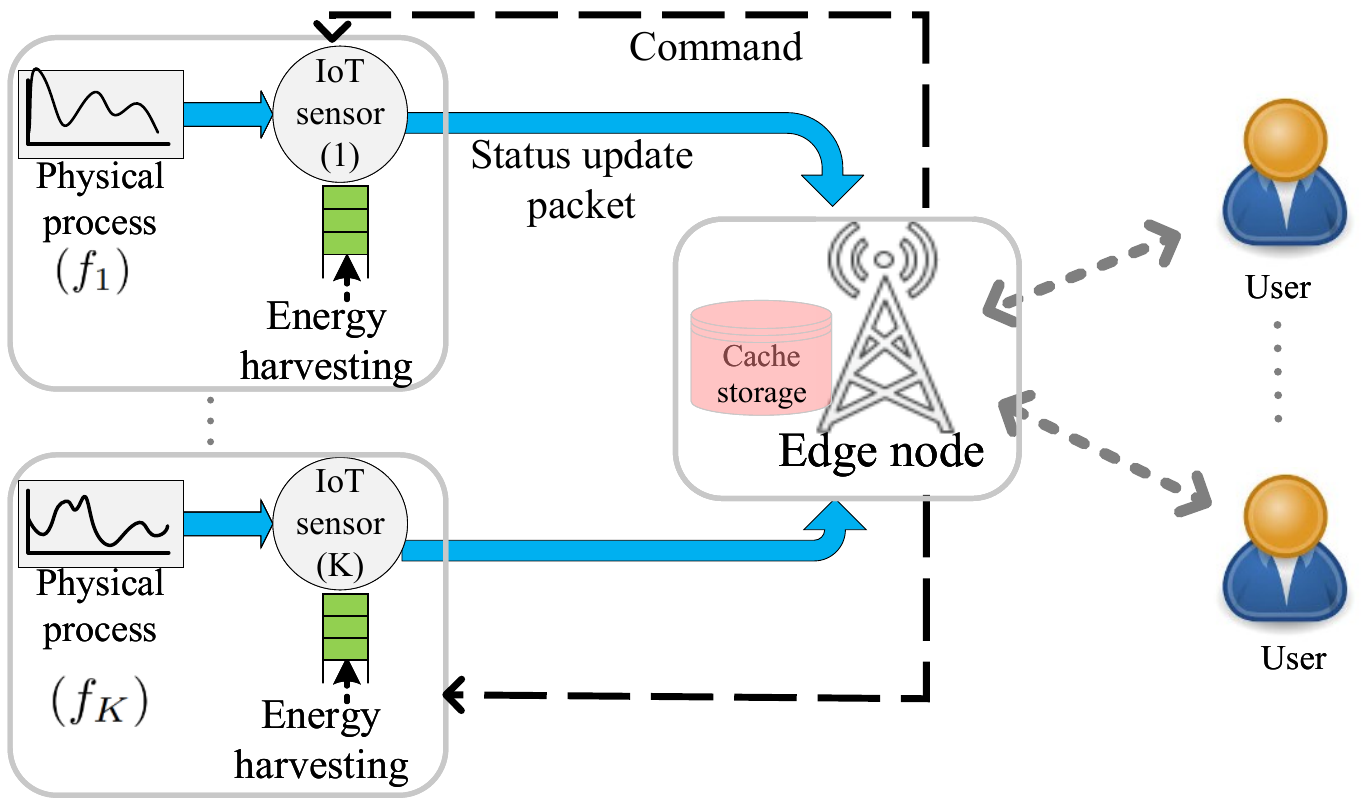}
\vspace{-2mm}
\caption{A multi-sensor IoT sensing network with $K$ EH sensors, an edge node, and users, which are interested in timely status update information of the physical processes measured by the sensors.
}
\label{fig_systemmodel}
\vspace{-5mm}
\end{figure}

We consider a multi-sensor status update system that consists of a set ${\mathcal{K}=\{1,\ldots,K\}}$ of $K$ energy harvesting (EH) sensors, an edge node (a gateway), and users, as depicted in Fig.~\ref{fig_systemmodel}. Users are interested in timely status information about random processes associated with physical quantities $f_k$, e.g., speed or temperature, each of which is independently measured by sensor ${k\in\mathcal{K}}$. 
The edge node provides an interface for the users to communicate with IoT sensors, i.e., the users receive the status updates only via the edge node. 

We consider a time-slotted system with slot indices ${t \in \mathbb{N}}$. At the beginning of slot $t$, users send requests for the status of physical quantities $f_k$ to the edge node. Let $r_{k}(t) \in \{0,1\}$, $t=1,2,\dots$, denote the random process of requesting the status of $f_k$  at the beginning of slot $t$; $r_{k}(t) = 1$ if the status  of $f_k$ is requested and  
$0$ otherwise.
The requests are independent across the sensors and time slots. Let $p_{k}$ be the probability that the status of $f_k$ is requested at a slot, i.e., $\mathrm{Pr}\{r_{k}(t) = 1\} = p_{k}$. 
We assume that all requests that arrive at the beginning of slot $t$ are handled by the edge node during the same slot. 

The edge node is equipped with a \textit{cache} 
that stores the most recently received status update packet from each sensor. Upon receiving a request for the status of $f_k$ at slot $t$, the edge node has two options to serve the request: 1) command sensor $k$ to send a fresh status update, or 2) use the previous measurement from the cache.
{We denote the \textit{command action of the edge node} at slot $t$ by $a_k(t) \in \{0,1\}$;}
$a_k(t)=1$ if the edge node commands sensor $k$ to send an update and 
$0$ otherwise.

We consider that, due to limited amount of radio resources (e.g.,
time-frequency resource blocks), at most ${N < K}$ sensors can transmit status updates to the edge node {within} each slot. This \textit{transmission constraint} imposes a limitation to the number of commands as
\begin{equation}\label{eq_bw_st}
    \textstyle\sum_{k=1}^{K}a_k(t) \leq N,\,\forall{t}.
\end{equation}
{We refer to $N$ as the \textit{transmission budget} hereinafter.}



\subsection{Energy Harvesting Sensors}\label{EH_model}
We assume that the sensors harvest energy from the environment. 
The energy arrivals at the sensors {are modeled} as independent Bernoulli processes $e_k(t) \in \left\lbrace 0 ,1\right\rbrace $, $t = 1,2,\dots$ {(see, e.g., \cite{Stamatakis2019control,pappas2020average})} with rates $\lambda_k$, $k\in\mathcal{K}$. 
Therefore, during each slot, sensor $k$ harvests one unit of energy with probability $\lambda_k$
and stores the energy in a battery with a finite capacity $B$.
We denote the battery level of sensor $k$ at the beginning of slot $t$ by $b_k(t)$, where ${b_k(t) \in \{0,\ldots,B\}}$.



We assume that measuring and transmitting a status update 
to the edge node consumes one unit of energy. 
Once sensor $k$ receives a command from the edge node (i.e., $a_k(t)=1$), it sends a status update if its battery is non-empty (i.e., $b_k(t) \geq 1$). {We denote the \textit{action of sensor $k$} at slot $t$ by $d_k(t) \in \left\lbrace 0 ,1\right\rbrace$;}
$d_k(t)=1$ if sensor $k$ sends a status update to the edge node and 
$0$ otherwise. Thus, 
\begin{equation}\label{eq_d}
d_k(t) =  a_k(t) \mathds{1}_{\{b_k(t) \geq 1\}},
\end{equation}
where $\mathds{1}_{\{\cdot\}}$ is the indicator function. Note that $d_k(t)$ in \eqref{eq_d} {determines} the energy {expenditure} of sensor $k$ at slot $t$. It is also worth noting that by \eqref{eq_d}, we have $d_k(t)\le{a_k(t)}$, and consequently, \eqref{eq_bw_st} implies that
$\sum_{k=1}^{K}d_k(t) \leq N$ for all slots; hence, the name transmission constraint for \eqref{eq_bw_st}. 
Finally, 
the evolution of the battery level of sensor $k$ is {given by}
\begin{equation}\label{eq_battery_evo}
b_k(t+1) = \min\left\lbrace  b_k(t)+e_k(t)-d_k(t) , B \right\rbrace.
\end{equation}




\subsection{Status Updating with Partial Battery Knowledge}\label{sec_system_model_partial_battery}
We model the practical operation mode of the network by considering that the edge node is informed about the sensors' battery level (only) via the received \textit{status update packets}. 
This is in stark contrast to the existing AoI-aware network designs (e.g., \cite{hatami2021spawc,Hatami2022AsymptoticallyOpt,hatami2022JointTcom,abd2019reinforcement}) which assume that the true battery levels are available at the edge node at each slot. 

Each status update packet sent by sensor $k$ consists of the measured value (status) of physical quantity $f_k$, a  time  stamp  representing  the time when the sample was generated, and the current battery level of the sensor. This leads to a situation where the edge node has only \textit{partial} knowledge about the battery level at each slot, i.e.,  \textit{outdated} knowledge based on the sensor’s last update.
Let {$\tilde{b}_k(t) \in \{1,2,\dots,B\}$} denote the \emph{knowledge} about the battery level of sensor $k$ at the edge node at slot $t$. 
At slot $t$, let $u_k(t)$ denote the most recent slot in which the edge node received a status update from sensor $k$, i.e., ${u_k(t) = \max \{t'| t'<t, d_k(t') = 1 \}}$.
Then, the true battery level and the knowledge about the level of sensor $k$ are interrelated as ${\tilde{b}_k(t) = b_k(u_k(t))}$.
In sequel, we 
refer to $\tilde{b}_k(t)$ 
simply as the \textit{partial battery knowledge}.

\subsection{On-demand Age of Information}\label{sec_AoI}
We use \textit{on-demand AoI} \cite{hatami2020aoi} to measure the freshness of information \textit{seen by the users} in our request-based status updating system. 
Let  $\Delta_k(t)$ be the AoI \cite{AoI_Orginal_12} about the physical quantity $f_k$ at the edge node at the beginning of slot  $t$, i.e., the number of slots elapsed since the generation of the most recently  received status update from sensor $k$.
Thus, the AoI about $f_k$ is a random process
$\Delta_k(t) \triangleq t - u_k(t)$.
We make a common assumption (e.g., \cite{Ceran2021MultiUserCMDP,ceran2021learningEH,leng2019AoIcognitive, Elvina2021SourceDiversityEH, Stamatakis2019control,hatami2020aoi,hatami2022JointTcom}) that $\Delta_k(t)$ is upper-bounded by a {sufficiently large} value $\Delta^{\mathrm{max}}$, i.e., $\Delta_k(t) \in \{1, 2,\ldots ,\Delta^{\mathrm{max}}\}$. 
At each slot, the AoI about $f_k$ drops to one if a status update from sensor $k$ is received and otherwise increases by one, i.e.,
\begin{equation}\label{eq_AoI}
\Delta_k(t+1)=
\begin{cases}
1,&\text{if} ~d_k(t)=1, \\
\min \{\Delta_k(t)+1,\Delta^{\mathrm{max}}\},&\text{if}~d_k(t)=0,
\end{cases}
\end{equation}
which is compactly written as $\Delta_k(t+1)=\min \{ (1-d_k(t)) \Delta_k(t)+1,\Delta^{\mathrm{max}}\}$.
We define on-demand AoI associated with sensor $k$ at slot $t$ as the sampled version of \eqref{eq_AoI} where the sampling is controlled by the request process $r_{k}(t)$, i.e.,
\begin{equation}\label{eq-on-demand-AoI}
\begin{array}{ll}
\Delta^\mathrm{OD}_{k}(t) &\hspace{-2mm} \triangleq r_{k}(t) \Delta_k(t+1) \\
&\hspace{-2mm} = r_{k}(t) \min \{ (1-d_k(t)) \Delta_k(t)+1,\Delta^{\mathrm{max}}\}.
\end{array}
\end{equation}
In \eqref{eq-on-demand-AoI}, since the requests come at the beginning of slot $t$ and the edge node sends {measurements} to the users at the end of the same slot, $\Delta_k(t+1)$ is the AoI about $f_k$ seen by the users.



\subsection{POMDP Construction} \label{sec_state_action_cost_def}
\subsubsection{State} Let ${s_k(t) \in \mathcal{S}_k}$ denote the state associated with sensor $k$
at slot $t$, defined as ${s_k(t) = (b_k(t),r_k(t), \Delta_k(t),\tilde{b}_k(t))}$;
$\mathcal{S}_k$ is the per-sensor state space with dimension ${ |\mathcal{S}_k| = 2B(B+1)\Delta^{\mathrm{max}}}$. 
We denote the \textit{observable} part of the state (visible by the edge node) by ${s_k^{\mathrm{v}}(t) = (r_k(t), \Delta_k(t),\tilde{b}_k(t))}$; thus, ${s_k(t) = (b_k(t),s_k^{\mathrm{v}}(t))}$. The state of the system at slot $t$ is expressed as $\mathbf{s}(t) = \left(s_1(t), \dots, s_K(t)\right) \in \mathcal{S}$, $\mathcal{S} = \mathcal{S}_1 \times \dots \times \mathcal{S}_K$; 
$|\mathcal{S}| = \prod_{k =1}^{K}2B(B+1)\Delta^{\mathrm{max}} = (2B(B+1)\Delta^{\mathrm{max}})^K$.

\subsubsection{Action}\label{sec-action-def}
The edge node decides at each slot whether to command sensor $k$ to send a fresh status update (and update the cache) or not, i.e., $a_k(t) \in \mathcal{A}_k =  \{0,1\}$, where $\mathcal{A}_k$ is the per-sensor action space. The action of the edge node at slot $t$ is given by a $K$-tuple
$\mathbf{a}(t) = \big(a_1(t), \dots, a_K(t)\big) \in \mathcal{A}$ with action space
$\mathcal{A}=\big\{(a_1,\ldots,a_K) \mid a_k\in \mathcal{A}_k
,\;\sum_{k=1}^{K}a_k\leq{N}\big\}$; $|\mathcal{A}| = \sum_{m = 0}^{N} \binom{K}{m}$.
Note that $\mathcal{A}$ considers the transmission constraint \eqref{eq_bw_st} in its definition. Additionally, we define the \textit{relaxed} action space that does not consider the transmission constraint \eqref{eq_bw_st} as $\mathcal{A}_{\mathrm{R}} = \mathcal{A}_1 \times \cdots \times \mathcal{A}_K = \{0,1\}^K$;  $|\mathcal{A}_{\mathrm{R}}| = 2^K$.
\subsubsection{Observation}
Let ${o_k(t) \in \mathcal{O}_k}$ be the edge node's \textit{observation} 
associated with sensor $k$ at slot $t$. We define it as the visible part of the state $s_k(t)$, i.e., ${o_k(t) = s_k^\mathrm{v}(t)}$. The observation space $\mathcal{O}$ has a finite dimension $ |\mathcal{O}| = 2B\Delta^\mathrm{max}$. The observation of the system is expressed as $\mathbf{o}(t) = (o_1(t),\dots,o_K(t)) \in \mathcal{O}$, $\mathcal{O} = \mathcal{O}_1,\dots,\mathcal{O}_K$; $|\mathcal{O}| = \prod_{k =1}^{K}2B\Delta^{\mathrm{max}} = (2B\Delta^{\mathrm{max}})^K$.
\subsubsection{Belief-state} 
{As the battery level in per-sensor state ${s_k(t) = (b_k(t),s_k^{\mathrm{v}}(t))}$ is not visible to the edge node, we introduce \textit{belief-states}, which preserve Markov property and are \textit{sufficient information states} \cite[Chapter~7]{sigaud2013markov} in respect to searching for an optimal policy.}
We define {the per-sensor belief-state}
at slot $t$ as ${z_k(t) = \left(\boldsymbol{\beta}_k(t), s_k^{\mathrm{v}}(t)\right) \in {\mathcal{Z}_k(t)}}$, where $\boldsymbol{\beta}_k(t)$ is \textit{belief} about the battery level
$b_k(t)$ and $\mathcal{Z}_k(t)$ is the per-sensor belief-state space; 
{the belief-state of the system at slot $t$ is $\mathbf{z}(t) = (z_1(t),\dots,z_K(t)) \in \mathcal{Z}$, $\mathcal{Z} = \mathcal{Z}_1 \times \dots \times \mathcal{Z}_k$.}
The per-sensor belief at slot $t$ is a $({B+1})$-dimensional vector ${\boldsymbol{\beta}_k(t) = (\beta_{k,0}(t),\dots,\beta_{k,B}(t))\tran{\in \mathcal{B}_k}}$, representing the probability distribution on the possible values of battery levels, where $\mathcal{B}_k \subset \mathbb{R}^{(B+1) \times 1}$ is the per-sensor belief space. 
Let $\phi_k^\mathrm{c}(t)$ be the \textit{complete information state} associated with sensor $k$ at slot $t$, which consists of an initial probability distribution over the states, and the complete \textit{history} of observations and actions up to slot $t$, i.e., $(o_k(1),\dots,o_k(t),a_k(1),\dots,a_k(t-1))$. 
The per-sensor belief $\boldsymbol{\beta}_k(t)$ represents the conditional probability distribution that the battery level of sensor $k$ has a certain value, given 
$\phi^\mathrm{c}_k(t)$. 
Thus, the entries of $\boldsymbol{\beta}_k(t)$ are defined as 
\begin{equation}\label{eq_def_update}
    \beta_{k,j}(t) \triangleq \Pr(b_k(t) = j \mid \phi^\mathrm{c}_k(t)),{~j \in \{0,1,\dots,B\}.}
\end{equation} 

The belief is updated at 
each slot based on the previous belief $\boldsymbol{\beta}_k(t)$, the current observation $o_k(t+1)$, and the previous action $a_k(t)$, i.e., $\boldsymbol{\beta}(t+1) = \tau_k(\boldsymbol{\beta}_k(t),o_k(t+1),a_k(t))$, where the belief update function $\tau_k(\cdot)$ 
is given by \cite[Prop.~1]{hatami2022spawc_partialbattery}.
\subsubsection{Policy}\label{sec_policy_def}
A policy $\pi$ is a rule that determines the action by observing the belief-state.
A randomized policy is a mapping from belief-state $\mathbf{z} \in \mathcal{Z}$ to a \textit{probability distribution} ${\pi(\mathbf{a}|\mathbf{z}) : \mathcal{Z} \times \mathcal{A} \rightarrow \left[ 0 , 1\right]}$, $\sum_{\mathbf{a} \in \mathcal{A}} \pi(\mathbf{a}|\mathbf{z}) = 1$,  of choosing each possible action $\mathbf{a} \in \mathcal{A}$. A deterministic policy is a special case where, in each state $\mathbf{z}$, $\pi(\mathbf{a}|\mathbf{z}) = 1$ for some $\mathbf{a}$; with a slight abuse of notation, we use $\pi(\mathbf{z})$ to denote the action taken in state $\mathbf{z}$  by a deterministic policy $\pi$.
In addition, we define a (relaxed) policy as $\pi_{\mathrm{R}}: \mathcal{Z} \times \mathcal{A}_{\R} \rightarrow \left[ 0 , 1\right]$ and a per-sensor policy as $\pi_k: \mathcal{Z}_k \times  \mathcal{A}_k \rightarrow \left[ 0 , 1\right]$.


\subsubsection{Cost Function}\label{sec_cost}
We define the cost associated with sensor $k$ at slot $t$ as the on-demand AoI for sensor $k$, i.e.,
\begin{equation}\label{persensor_cost}
c_k(t) 
= r_{k}(t) \Delta_k(t+1).
\end{equation}
\subsection{Problem Formulation}\label{subsec_problem_formulation}


For a given policy $\pi$, we define {the average cost as \textit{the average on-demand AoI over all sensors,}
i.e., }
\begin{equation}\label{eq_average_cost}
\bar{C}_{\pi} \triangleq \lim_{T\rightarrow\infty} \textstyle\frac{1}{KT}\sum_{t=1}^{T}\sum_{k=1}^{K} \mathbb{E}_{\pi}[c_{k}(t) \mid \mathbf{z}(0)],
\end{equation}
where $\mathbb{E}_{\pi}[\cdot]$ is the (conditional) expectation when the policy $\pi$ is applied to the system
and $\mathbf{z}(0) = \big(z_1(0),\ldots,z_K(0)\big)$ is the initial belief-state\footnote{{We assume that all policies $\pi$  induce a Markov chain with a single recurrent class plus a (possibly empty) set of transient states. Consequently}, the minimum average cost is independent of the initial state \cite[Chapter.~8]{puterman2014markov}.
}. We aim to find an optimal policy $\pi^\star$ that achieves the minimum average cost, i.e., 
\begin{equation}\label{average_cost_sp1}
(\textbf{P1})~~\pi^\star \in~ \textstyle{\argmin_{\pi}}~\bar{C}_{\pi}.
\end{equation}





{We can model (\textbf{P1}) as a POMDP and derive an optimal policy\footnote{{We can readily extend the approach in \cite{hatami2022spawc_partialbattery} for the multi-sensor scenario.}} $\pi^\star$.
Note, however, that the belief-state space $\mathcal{Z}$ and action space $\mathcal{A}$ grow exponentially in the number of sensors $K$, and thus, the complexity of finding an optimal policy grows exponentially in $K$, thereby being \textit{PSPACE-Hard}. Therefore, we next propose an asymptotically optimal low-complexity algorithm whose complexity increases only \textit{linearly} in $K$.}




\section{Relax-then-Truncate: {Asymptotically Optimal Status Updating}}\label{sec-lowcomplxity_alg}
We start by relaxing 
constraint \eqref{eq_bw_st} into a time average constraint and  model the \textit{relaxed problem} as a \textit{constrained POMDP} (CPOMDP). The CPOMDP problem is transformed into an unconstrained POMDP problem through the Lagrangian approach.
The POMDP problem decouples along the sensors; we find optimal per-sensor policies for a fixed Lagrange multiplier, whereas the optimal Lagrange multiplier is found via bisection. This procedure provides an optimal policy for the relaxed problem, called \textit{optimal relaxed policy}. 
As the final step, we propose an online \textit{truncation} procedure to ensure that constraint \eqref{eq_bw_st} is satisfied at each slot. Our 
analysis shows that this \textit{relax-then-truncate} approach is \textit{asymptotically optimal} as the number of sensors goes to infinity.

\subsection{CPOMDP Formulation}\label{sec_cmdp_modeling}
We define the average number of command actions under a policy $\pi_\R$ as
\begin{equation}\label{eq_average_transmision}
\bar{J}_{\pi_\R} \triangleq \lim_{T\rightarrow\infty}\textstyle\frac{1}{KT} \sum_{t=1}^{T} \sum_{k=1}^{K}\mathbb{E}_{\pi_\R}[a_k(t)],
\end{equation} 
and express the relaxed problem as
\begin{equation}
    \begin{array}{ll}
    (\textbf{P2})~~\pi_{\mathrm{R}}^\star \in \argmin_{\pi_{\mathrm{R}}} & \bar{C}_{\pi_\R}\\
    \hspace{1.8cm}
    \mbox{subject to} & \bar{J}_{\pi_\R} \leq \Gamma,
    \end{array}
    \label{st_opt2}
\end{equation}
where ${\Gamma \triangleq \frac{N}{K}}$ is the normalized transmission budget.
Note that the average cost obtained under $\pi_{\mathrm{R}}^\star$ is a \textit{lower bound} on the average cost obtained under $\pi^\star$, i.e.,
\begin{equation}\label{eq_lowerboundinequality}
 \bar{C}_{\pi^\star_\R} \leq \bar{C}_{\pi^\star}. 
\end{equation}

To solve (\textbf{P2}), we introduce a Lagrange multiplier $\mu$ and define the Lagrangian associated with problem (\textbf{P2}) as
\begin{equation}\label{lagrang_fcn}
    {\mathcal{L}}(\pi_{\mathrm{R}},\mu) \! \triangleq \!\! \lim_{T\rightarrow\infty} \textstyle\frac{1}{KT}\sum_{t=1}^{T}\sum_{k=1}^{K} \mathbb{E}_{\pi_{\mathrm{R}}}[c_k(t) + \mu a_k(t)] - \mu \Gamma.
\end{equation}
For a given ${\mu \geq 0}$, we define
the Lagrange dual function 
${\mathcal{L}^\star({\mu}) = \min_{\pi_{\mathrm{R}}} {\mathcal{L}}({\pi_{\mathrm{R}},\mu})}$. A policy that achieves ${\mathcal{L}}^\star({\mu})$ is called \textit{$\mu$-optimal}, denoted by $\pi_{\R,\mu}^\star$, and it is a solution of the following (unconstrained) POMDP problem
\begin{equation}\label{average_cost_relax_dual_sp2}
(\textbf{P3})~~\pi_{\R,\mu}^\star \in \textstyle\argmin_{\pi_{\mathrm{R}}} {\mathcal{L}}({\pi_{\mathrm{R}},\mu}).
\end{equation}


The optimal value of the CPOMDP problem (\textbf{P2}), $\bar{C}_{\pi^\star_{\R}}$, and the optimal value of the POMDP problem (\textbf{P3}), $\mathcal{L}^\star({\mu})$, ensures the following relation \cite[Corollary~12.2]{altman1999constrained}
\begin{equation}\label{dultiy-relation}
    \bar{C}_{\pi^\star_{\R}} = \textstyle\sup_{\mu \geq 0} 
    {\mathcal{L}}^\star({\mu}).
\end{equation}
Thus, an optimal policy for (\textbf{P2}) is found by a two-stage iterative algorithm:
1) for a given $\mu$, we find a $\mu$-optimal policy, and 2) we update $\mu$ in a direction that obtains  $\bar{C}_{\pi^\star_\R}$ according to \eqref{dultiy-relation}. 
These two steps are detailed in the following.

\subsubsection{Finding $\mu$-optimal Policy} 
\label{sec_opt_fixed_mu}
For a given $\mu$, (\textbf{P3}) is \textit{separable} across the sensors, i.e., (\textbf{P3}) decouples into $K$ per-sensor problems. To this end, we express the Lagrangian in
\eqref{lagrang_fcn} as $\mathcal{L}({\pi_{\mathrm{R}},\mu}) = \textstyle\frac{1}{K} \sum_{k=1}^{K}{\mathcal{L}_k}({\pi_k,\mu}) - \mu \Gamma$,
where 
\begin{equation}\label{per-sensor-lagrang-fcn}
    {\mathcal{L}}_k({\pi_k,\mu}) \!\triangleq\!\! \lim_{T\rightarrow\infty} \frac{1}{T}\textstyle\sum_{t=1}^{T} \mathbb{E}_{\pi_k}[c_k(t) + \mu a_k(t)],~k = 1, \dots, K.\notag
\end{equation}
Thus, finding an optimal policy $\pi_{\R,\mu}^\star$ reduces to finding $K$ per-sensor optimal policies, denoted by $\pi_{\R,\mu,k}^\star$, $k \in \mathcal{K}$, as
\begin{equation}\label{_persensor-average_cost_relax_dual_sp2}
(\textbf{P4})~~\pi_{\R,\mu,k}^\star \in \textstyle\argmin_{\pi_k} {\mathcal{L}}_{k}({\pi_k,\mu}),~k = 1,\dots,K.
\end{equation}
Each sub-problem (\textbf{P4})
can be modeled as an (unconstrained) POMDP problem. Particularly, the POMDP model associated with sensor $k$ is defined as the tuple $(\mathcal{S}_k,\mathcal{O}_k,\mathcal{A}_k,\Pr(s_k(t+1)| s_k(t), a_k(t)),\Pr(o_k(t)| s_k(t), a_k(t-1)),c_k(s_k(t),a_k(t)))$ \cite[Chap.~7]{sigaud2013markov}, 
where $\mathcal{S}_k$, $\mathcal{O}_k$, and $\mathcal{A}_k$ were defined in Section~\ref{sec_state_action_cost_def}, the state transition probabilities ${\Pr(s_k(t+1)|s_k(t),a_k(t))}$ are calculated as shown in {\cite[Section~III, Eq.~5]{hatami2022spawc_partialbattery}}, the observation function is given by $\text{Pr} (o_k(t)| s_k(t),a_k(t-1)) = \mathds{1}_{\{ o_k(t) = s_k^{\mathrm{v}}(t)\}}$, and the cost function is 
$c_k(s_k(t), a_k(t))+\mu a_k(t)$, where $c_k(s_k(t), a_k(t))$ is calculated using \eqref{persensor_cost}.
%
%
By {\cite[Theorem~1]{hatami2022spawc_partialbattery}}, the optimal average cost achieved by $\pi^\star_{\R,\mu,k}$, denoted by $\mathcal{L}_k^\star(\mu)$ (i.e., $\mathcal{L}_k^\star({\mu}) \triangleq \min_{\pi_k}{\mathcal{L}}_{k}({\pi_k,\mu})$)
satisfies the following equations
\begin{equation}\label{eq_persensor_ballman}
  \!\!\!\!\mathcal{L}_{k}^\star({\mu}) + h_{\R,\mu,k}(z)  = \textstyle\min_{a \in \mathcal{A}_k} Q_{\R,\mu,k}(z,a), z \in \mathcal{Z}_k.
\end{equation}
where $h_{\R,\mu,k}(z)$ is a relative value function, and $Q_{\R,\mu,k}(z, a)$ is an action-value function, which, for (per-sensor) belief-state $z = (\boldsymbol{\beta},r,\Delta,\tilde{b}) \in \mathcal{Z}_k$ and action $a \in \{0,1\}$, is given by
\begin{subequations}\label{eq_bellman_q_pomdp}
\begin{align}
    &\hspace{-1mm} Q_{\R,\mu,k}(z , 0)  = r \min\{\Delta+1,\Delta^{\mathrm{max}}\} + \textstyle\sum_{r^\prime = 0}^{1} [r^\prime p + \notag \\
    &(1-r^\prime)(1-p)] h(\boldsymbol{\Lambda}_k \boldsymbol{\beta},r^\prime,\min\{\Delta+1,\Delta^{\mathrm{max}}\}, \tilde{b}), \\
    & \hspace{-1mm} Q_{\R,\mu,k}(z ,1)  = [r \beta_{0}  \min\{\Delta+1,\Delta^{\mathrm{max}}\} + r(1-\beta_{0})] + \beta_0\notag \\ 
    & \hspace{-1mm} \textstyle\sum_{r^\prime = 0}^{1}
    [r^\prime p + (1-r^\prime)(1-p)] h(\boldsymbol{\rho}_k^{0},r^\prime,\min\{\Delta+1,\Delta^{\mathrm{max}}\},\tilde{b})\notag \\
    &\hspace{-1mm} +\textstyle\sum_{j = 1}^{B} \beta_{j} \big[p h(\boldsymbol{\rho}_k^{j},1,1,j) + (1-p) h(\boldsymbol{\rho}_k^{j},0,1,j) \big],
\end{align}
\end{subequations}
{where a left stochastic matrix $\boldsymbol{\Lambda}_k$ and vectors ${\boldsymbol{\rho}_k^{0},\boldsymbol{\rho}_k^{1},\ldots,\boldsymbol{\rho}_k^{B}}$ are constructed as shown in \cite[Proposition~1]{hatami2022spawc_partialbattery}.}
Further, an optimal action in belief-state ${z \in \mathcal{Z}_k}$ is given by
\begin{equation}\label{eq_optimal_policy}
 \pi_{\R,\mu,k}^\star(z) = \textstyle\argmin_{a\in\mathcal{A}_k} Q_{\R,\mu,k}(z,a),~ z\in \mathcal{Z}_k.  
\end{equation}

An optimal policy $\pi_{\R,\mu,k}^\star$ can be found by converting the Bellman's optimality equation \eqref{eq_bellman_q_pomdp} into an iterative procedure, called relative value iteration algorithm (RVIA) \cite[Section~8.5.5]{puterman2014markov}. At each iteration $i = 0,1,\ldots$, we {first} update $Q_{\R,\mu,k}^{(i+1)}(z,a)$ by using $h_{\R,\mu,k}^{(i)}(z)$ in \eqref{eq_bellman_q_pomdp}, and then 
\begin{equation}\label{eq_v_itr}
\begin{array}{ll}
&V_{\R,\mu,k}^{(i+1)}(z) =\min_{a\in\mathcal{A}} Q_{\R,\mu,k}^{(i+1)}(z,a), \\& h_{\R,\mu,k}^{(i+1)}(z) = V_{\R,\mu,k}^{(i+1)}(z) - V_{\R,\mu,k}^{(i+1)}(z_{\mathrm{ref}}),
\end{array}
\end{equation}
where ${z_{\mathrm{ref}} \in \mathcal{Z}_k}$ is an arbitrary reference state.

While the sequences in \eqref{eq_v_itr} converge (regardless of the initialization $h_{\R,\mu,k}^{(0)}(z)$), finding $V_{\R,\mu,k}(z)$ (and $h_{\R,\mu,k}(z)$) iteratively via \eqref{eq_v_itr} is intractable, because the belief space $\mathcal{B}_k$ has infinite dimension. As a solution, we exploit a specific pattern in the beliefs' evolution that allows to truncate the belief space $\mathcal{B}_k$ {into a \textit{finite} belief space $\hat{\mathcal{B}}_k$} and subsequently develop a practical iterative algorithm relying on \eqref{eq_v_itr}, {as detailed in \cite[Sect.~IV]{hatami2022spawc_partialbattery}}. {Moreover, by \cite[Theorem~2]{hatami2022spawc_partialbattery}, $V_{\R,\mu,k}(\cdot)$ is fixed with respect to $\tilde{b}$, and consequently, it does not have any effect on $\pi^\star_{\R,\mu,k}$ in \eqref{eq_optimal_policy}. Thus, $\tilde{b}$ is removed from the belief-state in the algorithm.}
{The proposed iterative algorithm that finds $\mu$-optimal policies is presented in Algorithm~\ref{alg_cmdp_RVIA}~(Lines~\ref{alg-lst-line-function-RVIA}--\ref{alg-last-line-function-RVIA}).}
\subsubsection{Finding the Optimal Lagrange Multiplier}\label{sec_opt_mu}
{Note that}
$\bar{C}_{\pi_{\R,\mu}^\star}$ and ${\mathcal{L}}(\pi_{\R,\mu}^\star,\mu)$ are increasing in $\mu$, whereas $\bar{J}_{\pi_{\R,\mu}^\star}$ is decreasing in $\mu$
\cite[Lemma~3.1]{beutler1985optimal}. Therefore, we seek for 
the smallest value of the Lagrange multiplier such that $\pi_{\R,\mu}^\star$ satisfies the average transmission constraint in \eqref{st_opt2}.
We define the optimal Lagrange multiplier as \cite{beutler1985optimal}
\begin{equation}\label{eq-mu-star}
    \mu^* \triangleq \inf \big\{\mu \geq 0 \mid \bar{J}_{\pi_{\R,\mu}^\star} \leq \Gamma\big\},
\end{equation}
where $\bar{J}_{\pi_{\R,\mu}^\star}$ is the average number of command actions under $\pi_{\R,\mu}^\star$.
From \eqref{eq_average_transmision}
and the fact that (\textbf{P3}) decouples across the sensors, $\bar{J}_{\pi_{\R,\mu}^\star}$ is calculated as $ \bar{J}_{\pi^\star_{\R,\mu}}  = \frac{1}{K}\sum_{k = 1}^{K}\bar{J}_{\pi_{\R,\mu,k}^\star}$, where $\bar{J}_{\pi_{\R,\mu,k}^\star}$ denotes the per-sensor average number of command actions under $\pi_{\R,\mu,k}^\star$, which is defined as
\begin{equation}\label{eq_persensor_average_transmision}
\bar{J}_{\pi_{\R,\mu,k}^\star} \triangleq 
\lim_{T\rightarrow\infty}\textstyle\frac{1}{T}\sum_{t=1}^{T} \mathbb{E}_{\pi_{\R,\mu,k}^\star}[a_k(t)].
\end{equation} 

We now characterize an optimal relaxed policy $\pi^\star_{\R}$ for (\textbf{P2}). If 
$\frac{1}{K}\sum_{k =1}^{K}\bar{J}_{\pi_{\R,\mu^*,k}^\star} = \Gamma$, then, $\pi_{\R,\mu^*,k}^\star$, ${k\in \mathcal{K}}$, form an optimal policy for (\textbf{P2}), i.e., $\pi_{\R}^\star = \pi_{\R,\mu^*}^\star$. 
Otherwise, $\pi_{\R}^\star$ is a mixture of two deterministic policies $\pi_{\R,\mu^{*-}}^\star$ and $\pi_{\R,\mu^{*+}}^\star$, which are defined by \cite[Theorem~4.4]{beutler1985optimal}
\begin{equation}\label{eq-mixingpolicies_def}
    \pi_{\R,\mu^{*-}}^\star  \triangleq \lim_{\mu \rightarrow \mu^{*-}} \pi_{\R,\mu}^\star ~\mathrm{and}~
    \pi_{\R,\mu^{*+}}^\star \triangleq \lim_{\mu \rightarrow \mu^{*+}} \pi_{\R,\mu}^\star,
\end{equation}
and is written symbolically as $\pi^\star_{\R} \triangleq \eta \pi_{\R,\mu^{*-}}^\star + (1-\eta) \pi_{\R,\mu^{*+}}^\star$, where $\eta$ is the mixing factor. This mixed policy is a  stationary randomized policy where the action at each belief-state $\mathbf{z}$ is $\pi_{\R,\mu^{*-}}^\star(\mathbf{z})$ with probability $\eta$ and $\pi_{\R,\mu^{*+}}^\star(\mathbf{z})$ with probability $1 -\eta$, where $\eta$ is obtained\footnote{{As there is no closed-form for $\eta\! \in\! [0,1]$, 
numerical search is used.
}} such that $\bar{J}_{\pi^\star_{\R}} = \Gamma$.


To search for $\mu^*$ as defined in \eqref{eq-mu-star}, we apply {bisection} that exploits the monotonicity of $\bar{J}_{\pi_{\R,\mu}^\star}$ with respect to $\mu$.
Particularly, if $\frac{1}{K}\sum_{k =1}^{K}\bar{J}_{\pi_{\R,\mu,k}^\star} \leq \Gamma$ for ${\mu = 0}$, then the constraint in \eqref{st_opt2} is inactive, and an optimal policy for (\textbf{P2}) is $\pi^\star_{\mathrm{R},0}$. Otherwise, we apply an iterative update procedure until ${|\mu^{+} - \mu^{-}| < \epsilon}$ and $\frac{1}{K}\sum_{k =1}^{K}\bar{J}_{\pi_{\R,\mu,k}^\star} \leq \Gamma$ are satisfied.
{Details are expressed in Algorithm~\ref{alg_cmdp_RVIA} (Lines~\ref{alg-1st-Line}--\ref{alg-Last-Line}).}


\subsection{{Truncation Procedure}}\label{sec_truncation}
When Algorithm~\ref{alg_cmdp_RVIA} has been executed, there is no guarantee that the per-slot constraint \eqref{eq_bw_st} is satisfied under optimal relaxed policy $\pi^\star_{\R}$. Thus, we propose the following truncation procedure that satisfies \eqref{eq_bw_st} at each slot.
At slot $t$, let ${\mathcal{X}(t) = \{k\mid a_k(t) = 1, k \in \mathcal{K}\} \subseteq \mathcal{K}}$ denote the set of sensors that are commanded under $\pi^\star_{\R}$. The truncation step separates into two cases: 1) if ${|\mathcal{X}(t)| \leq N}$, the edge node simply commands all sensors in $\mathcal{X}(t)$, and 2) otherwise, the edge node selects $N$ sensors from the set $\mathcal{X}(t)$ \textit{randomly (uniform)} and commands them to send status updates. 



\algrenewcommand\algorithmicrequire{\textbf{Input}}
\algrenewcommand\algorithmicensure{\textbf{Output}}

\begin{algorithm}[t!]
\begin{small}
\caption{Policy design for the CPOMDP problem (\textbf{P2})}\label{alg_cmdp_RVIA}
\begin{algorithmic}[1]
\State \textbf{Initialize} Set $\mu \gets 0$, $\mu^{-} \gets 0$, $\mu^{+}$ as a large positive number, and determine a small $\epsilon > 0$ \label{alg-1st-Line}
\State \Call{RVIA-POMDP}{$\mu$} \hspace{-0.5cm}\Comment{\textit{run function RVIA-POMDP for $\mu = 0$}}
\If{$\bar{J}_{\pi^\star_{\R,\mu}} \leq \Gamma$}
\State $\pi^\star_{\R} = \pi^\star_{\R,\mu}$
\Else
\While{$|\mu^{+} - \mu^{-}| > \epsilon$}
\State \hspace{-6.5mm}\Call{RVIA-POMDP}{$\frac{\mu^{+} \!+\! \mu^{-}}{2}$} \hspace{-2.3mm}\Comment{\textit{run RVIA-POMDP for $\mu \!\!=\! \frac{\mu^{+} \!+\! \mu^{-}}{2}$}}
\State \hspace{-2.8mm} \textbf{if} $\bar{J}_{\pi^\star_{\R,\mu}} \geq \Gamma$ \textbf{then} $\mu^{-} \gets \mu$ \textbf{else} $\mu^{+} \gets \mu$
\EndWhile
\State  $\mu^* \gets 1/2(\mu^{-}+\mu^{+})$, $\mu^{*-} \gets \mu^{-}$, and $\mu^{*+} \gets \mu^{+}$
\State\hspace{-3mm} \textbf{if} $\bar{J}_{\pi^\star_{\R,\mu}}\!\!\!\!\! =\! \Gamma$ \textbf{then} $\pi^\star_{\R}\! =\! \pi^\star_{\R,\mu^*}$ \textbf{else} $\pi^\star_{\R}\!\! \triangleq\! \eta \pi_{\R,\mu^{*-}}^\star\!\! + \!(1-\eta) \pi_{\R,\mu^{*+}}^\star$
\EndIf 
\State {{\textbf{Output}: optimal relaxed policy $\pi^\star_{\R}$}}
\label{alg-Last-Line}
\vspace*{-.5\baselineskip}\Statex\hspace*{\dimexpr-\algorithmicindent-2pt\relax}\rule{\columnwidth}{0.4pt}
\Function{RVIA-POMDP}{$\mu$}\label{alg-lst-line-function-RVIA}
\Comment{\textit{find optimal policies $\pi^\star_{\R,\mu,k}$ for fixed $\mu$ (i.e., $\mu$-optimal policies)}} \vspace{1.5mm}
\State \hspace{-5mm}\textbf{Initialize} $V_{\R,\mu,k}(z) \leftarrow 0$, $h_{\R,\mu,k}(z) \leftarrow 0$,$~\forall z = \{\boldsymbol{\beta},r,\Delta\}, \boldsymbol{\beta}\in \hat{\mathcal{B}}, r \in  \{0,1\}, \Delta \in \{1,\dots,\Delta^{\mathrm{max}}\}$, determine an arbitrary ${z_{\textrm{ref}} \in \mathcal{Z}}$ and a small threshold ${\theta > 0}$
\State\hspace{-3mm} \textbf{for} $k = 1, \dots, K$ \textbf{do}
\Repeat
\State \hspace{-4mm} calculate $Q_{\R,\mu,k}(z,0)$ and $Q_{\R,\mu,k}(z,1)$ by \eqref{eq_bellman_q_pomdp}, for all $z$
\State \hspace{-4mm} $V_{\mathrm{tmp}}(z) \leftarrow \min_{a \in \mathcal{A}} Q(z,a)$, for all $z$
\State \hspace{-4mm}$\delta\! \gets\! \max_{z} (V_{\textrm{tmp}}(z) - V_{\R,\mu,k}(z)) - \min_{z} (V_{\textrm{tmp}}(z) - V_{\R,\mu,k}(z))$
\State \hspace{-4mm}$V_{\R,\mu,k}(z) \leftarrow V_{\textrm{tmp}}(z)$,  for all $z$
\State \hspace{-2.5mm} $h_{\R,\mu,k}(z) \leftarrow V_{\R,\mu,k}(z) - V_{\R,\mu,k}(z_{\mathrm{ref}})$,  for all $z$
\Until{$\delta < \theta$}
\State $\pi^\star_{\R,\mu,k}(z) = \argmin_{a \in \mathcal{A}} Q(z,a)$, for all $z$
\State \hspace{-3mm} \textbf{end for}
\State \hspace{-3mm}{\textbf{Output}: per-sensor optimal policies $\pi^\star_{\R,\mu,k}, k \in \mathcal{K}$}
\EndFunction \label{alg-last-line-function-RVIA}
\end{algorithmic}
\end{small}
\end{algorithm}

\subsection{Asymptotic Optimality of Relax-then-Truncate Approach}\label{sec_optimality}
We next analyze the optimality of the relax-then-truncate policy, which is denoted by $\tilde{\pi}$.
\vspace{-2mm}
\begin{lemm}\label{lemma_STD}
{Denoting the standard deviation of a random variable $X$ by $\mathrm{STD}(X)$, we have $\mathrm{STD}(|\mathcal{X}(t)|) \leq \sqrt{K}$.} 
\end{lemm}
\vspace{-2mm}
\textit{Proof}. 
The cardinality of set $\mathcal{X}(t)$ (i.e., the set of sensors that are commanded under $\pi_{\R}^\star$) can be written as ${|\mathcal{X}(t)| = \sum_{k = 1}^{K} a_k(t)}$, where ${a_k(t)\in \{0,1\}}$, ${k\in \mathcal{K}}$, are $K$ independent binary random variables. Therefore, random variable $|\mathcal{X}(t)|$ has a Poisson binomial distribution. Let $\omega_k(t)$ be the probability that sensor $k$ is commanded at slot $t$ under policy $\pi_{\R}^\star$, i.e., $\omega_k(t) \triangleq \Pr(a_k(t) = 1)$. Thus, we have
\begin{equation}\notag
    \mathrm{STD}(|\mathcal{X}(t)|) = \sqrt{\textstyle\sum_{k = 1}^{K} \underbrace{\omega_k(t) (1-\omega_k(t))}_{\leq 1}} {\leq} \sqrt{K}.
\end{equation}
\vspace{-3mm}
\begin{lemm}\label{lemma_STD_VAR_relatiopn}
{Denoting the Mean Absolute Deviation of a random variable $X$ by $\mathrm{MAD}(X)$, we have $ \mathrm{MAD}(X) \leq \mathrm{STD}(X)$.}
\end{lemm}
\vspace{-4mm}
\begin{proof}
Applying the Jensen's inequality for the convex function $f(\cdot) = (\cdot)^2$, i.e., $f(\mathbb{E}[\cdot]) \leq \mathbb{E}[f(\cdot)]$, we have
\begin{align}
& \underbrace{(\mathbb{E}[|X- \mathbb{E}(X)|])^2}_{ = (\mathrm{MAD}(X))^2} \leq \underbrace{\mathbb{E}[|X - \mathbb{E}(X)|^2]}_{= (\mathrm{STD}(X))^2},
\end{align}
which implies 
${\mathrm{MAD}(X) \leq \mathrm{STD}(X)}$.
\end{proof}
\vspace{-3mm}
\begin{theorem}\label{theorem-asymptotically-opt}
For any {normalized transmission budget} $\Gamma > 0$, the relax-then-truncate policy $\tilde{\pi}$ is asymptotically optimal with respect to the number of sensors, i.e., ${\lim_{K\rightarrow\infty} (\bar{C}_{\tilde{\pi}} - \bar{C}_{\pi^\star}) = 0}$.
\end{theorem}
\vspace{-2mm}
\begin{proof}
Let ${\mathcal{T}(t) \subset \mathcal{X}(t)}$ denote the set of \textit{truncated} sensors at slot $t$, i.e., the sensors that are not commanded under the relax-then-truncate policy $\tilde{\pi}$ while they are commanded under policy $\pi_{\mathrm{R}}^\star$. 
By the truncation procedure, if ${|\mathcal{X}(t)|> N}$, $N$ sensors are chosen (uniform) randomly from the set $\mathcal{X}(t)$ and commanded; ${|\mathcal{X}(t)| - N}$ sensors are not commanded. The probability that sensor $k$ belongs to $\mathcal{T}(t)$ is $\mathds{1}_{\{|\mathcal{X}(t)|>N|\}}\left(\frac{|\mathcal{X}(t)|-N}{|\mathcal{X}(t)|}\right)$. At each slot, the additional per-sensor cost under $\tilde{\pi}$ compared to $\pi^\star_{\R}$ is at most $\Delta^{\mathrm{max}}$ (see \eqref{persensor_cost}). Therefore, the expected additional cost over all sensors under $\tilde{\pi}$ compared to $\pi^\star_{\R}$ is upper bounded by
${K\Delta^{\mathrm{max}} {\frac{(|\mathcal{X}(t)| - N)^{+}}{|\mathcal{X}(t)|}}}$,
where $(\cdot)^+ \triangleq \max \{0,\cdot\}$.

We introduce the following (penalized) {strategy} $\hat{\pi}_{\mathrm{R}}$: at each slot, command the sensors based on $\pi^\star_{\R}$ but add a penalty $K\Delta^{\mathrm{max}} \frac{(|\mathcal{X}(t)| - N)^{+}}{|\mathcal{X}(t)|}$ to the cost over all sensors.
Clearly, the average cost obtained under $\hat{\pi}_{\mathrm{R}}$ is not less than that obtained by $\tilde{\pi}$, i.e., ${\bar{C}_{\tilde{\pi}} \leq \bar{C}_{\hat{\pi}_{\mathrm{R}}}}$. Also, recall from \eqref{eq_lowerboundinequality} that 
${\bar{C}_{\pi^\star_\R} \leq  \bar{C}_{\pi^\star}}$. Finally, policy $\tilde{\pi}$ is a sub-optimal solution for  (\textbf{P1}), i.e., ${\bar{C}_{\pi^\star} \leq \bar{C}_{\tilde{\pi}}}$. To conclude, we have 
\begin{equation}\label{eq_policy_order_final}
\bar{C}_{\pi^\star_\R} \leq  \bar{C}_{\pi^\star} \leq \bar{C}_{\tilde{\pi}} \leq \bar{C}_{\hat{\pi}_\R}.
\end{equation}

Using \eqref{eq_policy_order_final}, the difference between the average cost obtained by the proposed relax-then-truncate policy $\tilde{\pi}$ and the average cost obtained by an optimal policy $\pi^\star$ is upper bounded as
\begin{equation}\notag
\begin{array}{ll}
    &\bar{C}_{\tilde{\pi}} - \bar{C}_{\pi^\star} \overset{(a)}{\leq} \bar{C}_{\hat{\pi}_\R} - \bar{C}_{\pi^\star_\R}
    \\& =\lim_{T\rightarrow\infty} \frac{1}{KT} \sum_{t=1}^{T}\mathbb{E}_{\pi^\star_\R}\left[K\Delta^{\mathrm{max}} \frac{(|\mathcal{X}(t)| - N)^{+}}{|\mathcal{X}(t)|} \right]
    \\& \overset{(b)}{\leq} \frac{\Delta^{\mathrm{max}}}{N} \lim_{T\rightarrow\infty}\frac{1}{T} \sum_{t=1}^{T} \mathbb{E}_{\pi^\star_\R}\left[(|\mathcal{X}(t)| - N)^+\right]\\
    & \overset{(c)}{\leq} \frac{\Delta^{\mathrm{max}}}{N}  \lim_{T\rightarrow\infty}\frac{1}{T}\sum_{t=1}^{T} \mathbb{E}_{\pi^\star_\R}\left[(|\mathcal{X}(t)| - \mathbb{E}_{\pi^\star_\R}[|\mathcal{X}(t)|])^+\right]\\
    & \overset{(d)}{\leq} \frac{\Delta^{\mathrm{max}}}{N}  \lim_{T\rightarrow\infty}\frac{1}{T}\sum_{t=1}^{T} \underbrace{\mathbb{E}_{\pi^\star_\R}\Big[\big||\mathcal{X}(t)| - \mathbb{E}_{\pi^\star_\R}[|\mathcal{X}(t)|]\big|\Big]}_{= \mathrm{MAD}(|\mathcal{X}(t)|)}
    \end{array}
\end{equation}
\begin{equation}
    \begin{array}{ll}
    & \overset{(e)}{\leq} \frac{\Delta^{\mathrm{max}}}{K\Gamma}  \lim_{T\rightarrow\infty}\frac{1}{T}\sum_{t=1}^{T} \mathrm{STD}(|\mathcal{X}(t)|)\overset{(f)}{\leq} \frac{\Delta^{\mathrm{max}}}{\Gamma\sqrt{K}}\notag
\end{array}
\end{equation}
where $(a)$ follows from \eqref{eq_policy_order_final}, $(b)$ follows from $\frac{(|\mathcal{X}(t)| - N)^+}{|\mathcal{X}(t)|} \leq \frac{(|\mathcal{X}(t)| - N)^+}{N}$, $(c)$ follows from $\mathbb{E}_{\pi^\star_{\R}}[|\mathcal{X}(t)|] \leq N$, {for sufficiently large $t$}, $(d)$ follows from $(\cdot)^{+} \leq |\cdot|$, {$(e)$ follows from Lemma~\ref{lemma_STD_VAR_relatiopn}, 
and $(f)$ follows from Lemma~\ref{lemma_STD}. Therefore, we have 
${\lim_{K\rightarrow\infty} (\bar{C}_{\tilde{\pi}} - \bar{C}_{\pi^\star}) = 0}$, {which concludes the proof.}}
\end{proof}
\section{Numerical Results}\label{sec_simulation}
We consider a scenario where ${p_k = 0.8}$, ${\Delta^{\mathrm{max}} = 64}$, and ${B = 3}$.
{Each sensor is assigned an energy harvesting rate $\lambda_k$ from the set $\{0.01,0.02,\dots, 0.1\}$ sequentially:} sensors ${1,11,\ldots}$ have the rate $0.01$, sensors ${2,12,\ldots}$ have the rate $0.02$, and so on.
%
The following benchmarks are used for comparison.
1) A (request-aware) greedy policy where the edge node commands at most $N$ sensors with the largest AoI from the set $\mathcal{W}(t) \triangleq \{k \mid r_k(t) = 1, k\in\mathcal{K}\}$, i.e., the set of sensors whose status are requested by a user,
2) The lower bound, obtained by following an optimal relaxed policy $\pi^\star_{\R}$ (see \eqref{eq_lowerboundinequality}), and 
3) the case where the edge node knows the exact battery levels at each slot and which uses the relax-then truncate approach to find an asymptotically optimal policy  \cite{hatami2022JointTcom}.

\newcommand\fw{.8}
\begin{figure}[t!]
\centering
\subfigure [$\Gamma = 0.02$]{%
\includegraphics[width=\fw \columnwidth]{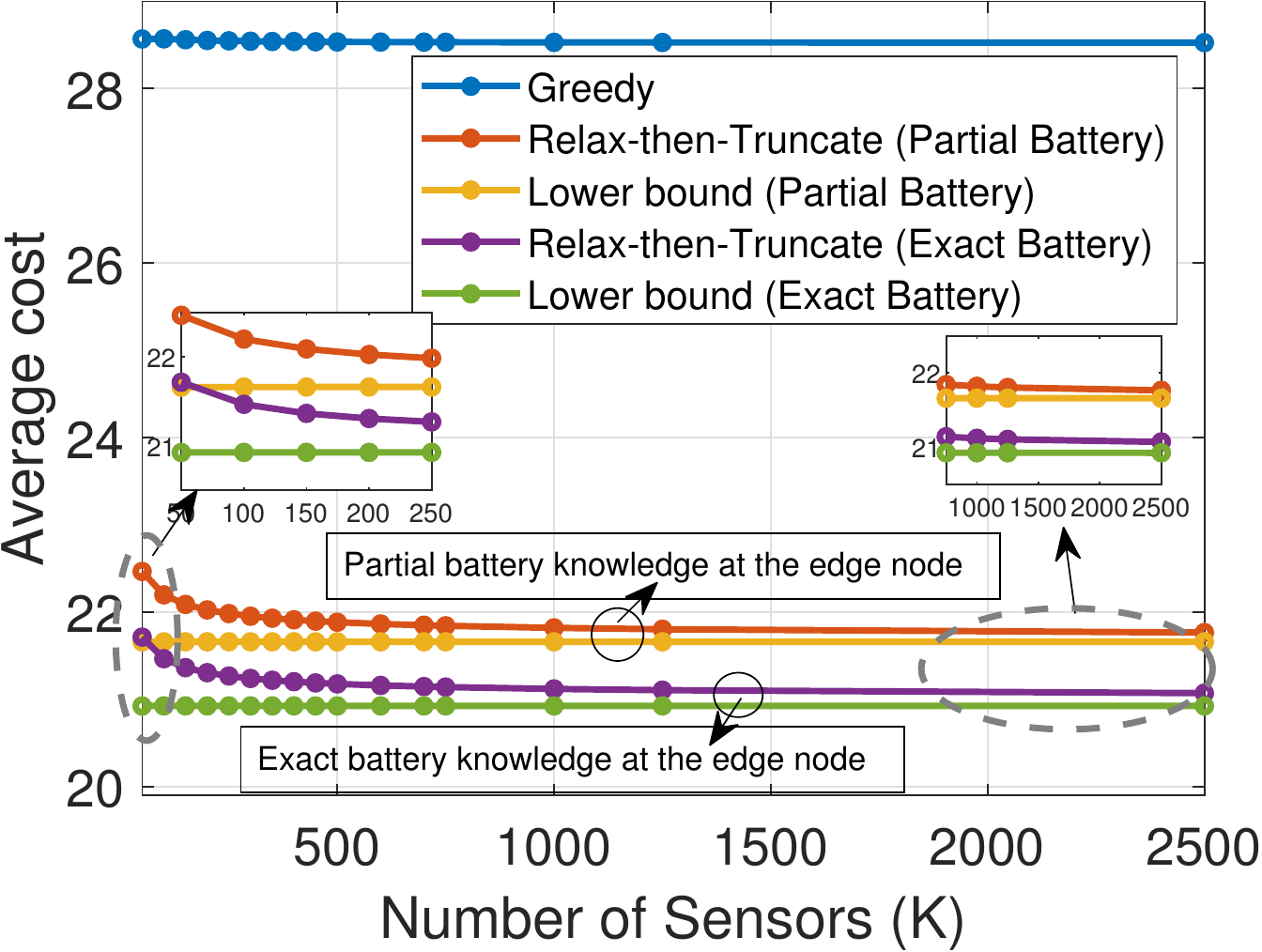}
}\vspace{-2mm}
\subfigure [$\Gamma = 0.15$]{%
\includegraphics[width=\fw \columnwidth]{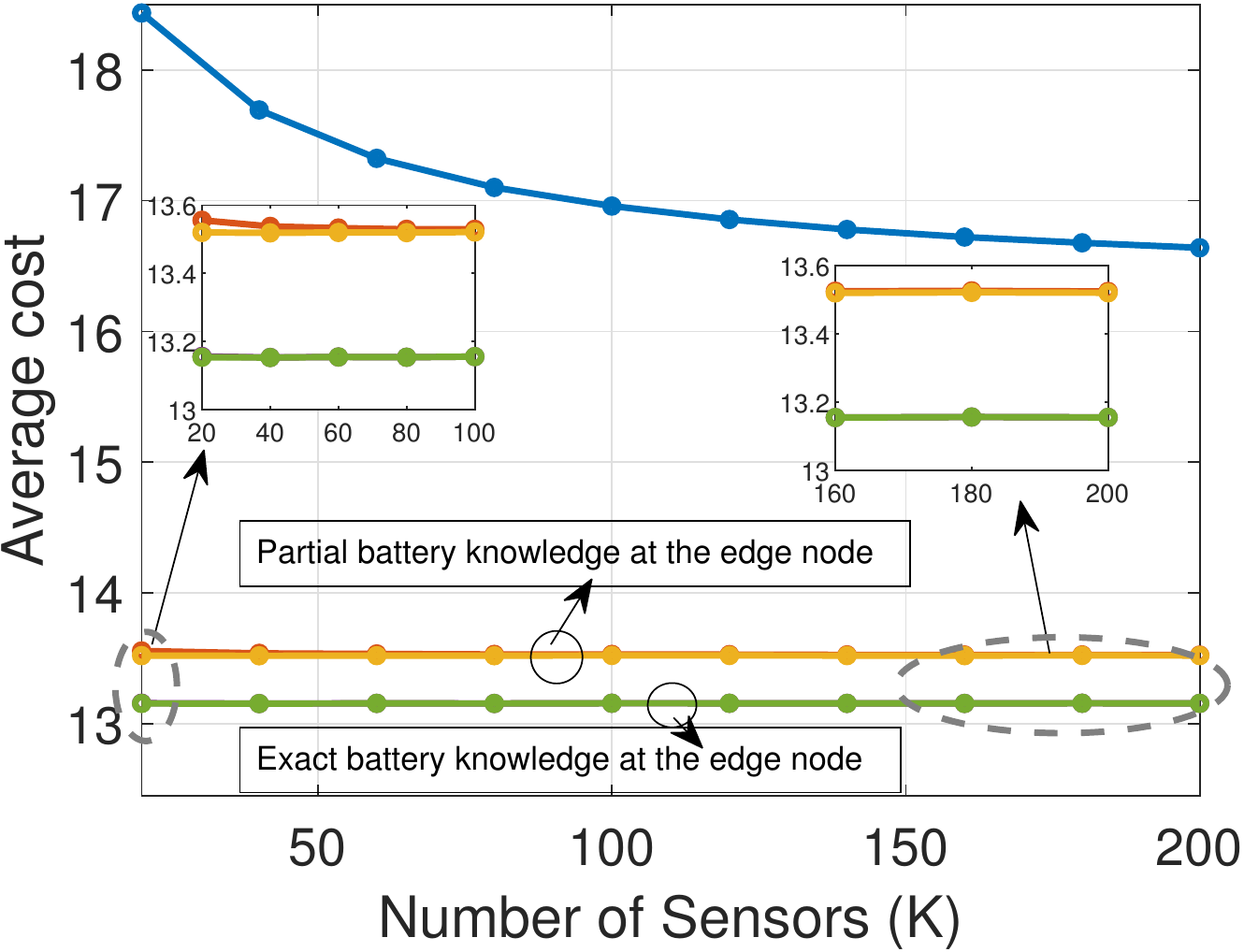}%
}
\vspace{-4mm}
\caption{Performance of the proposed {relax-then-truncate approach} in terms of average cost with respect to the number of sensors $K$.
}\vspace{-5mm}
\label{fig_perf_alpha_fixed}
\end{figure}

$\text{Fig.\ \ref{fig_perf_alpha_fixed}}$ depicts the performance of the relax-then-truncate algorithm with respect to the number of sensors $K$ for different values of normalized transmission budget $\Gamma$. The results are obtained by averaging each algorithm over $10$ episodes, each of length $10^7$ slots.
{First, the proposed algorithm reduces the average cost by approximately $30~\%$ compared to the greedy policy.}
Due to asymptotic optimality of the proposed algorithm, 
the gap between the proposed policy and the lower bound is very small for large values of $K$; the same holds true for the exact battery knowledge (see also \cite{hatami2022JointTcom}).
Interestingly, both relax-then-truncate approaches perform close to the optimal solutions even for moderate numbers of sensors.
{Moreover, \mbox{$\text{Figs.\ \ref{fig_perf_alpha_fixed}(a) and (b)}$} show that for large $\Gamma$, the proposed policy approaches the optimal performance for smaller values of $K$.}
This is because the proportion of the sensors that can be commanded 
at each slot increases as $\Gamma$ increases, and thus, the proportion of truncated sensors (i.e., those that are not commanded under $\tilde{\pi}$ compared to $\pi^\star_{\R}$) decreases.
{Furthermore, the performance of the proposed approach is not too far from the performance under the exact battery knowledge; this relatively small gap shows the impact of the uncertainty about the sensors' battery levels.}

\begin{figure}[t]
\centering
\subfigure[]{
\includegraphics[width=\fw \columnwidth]{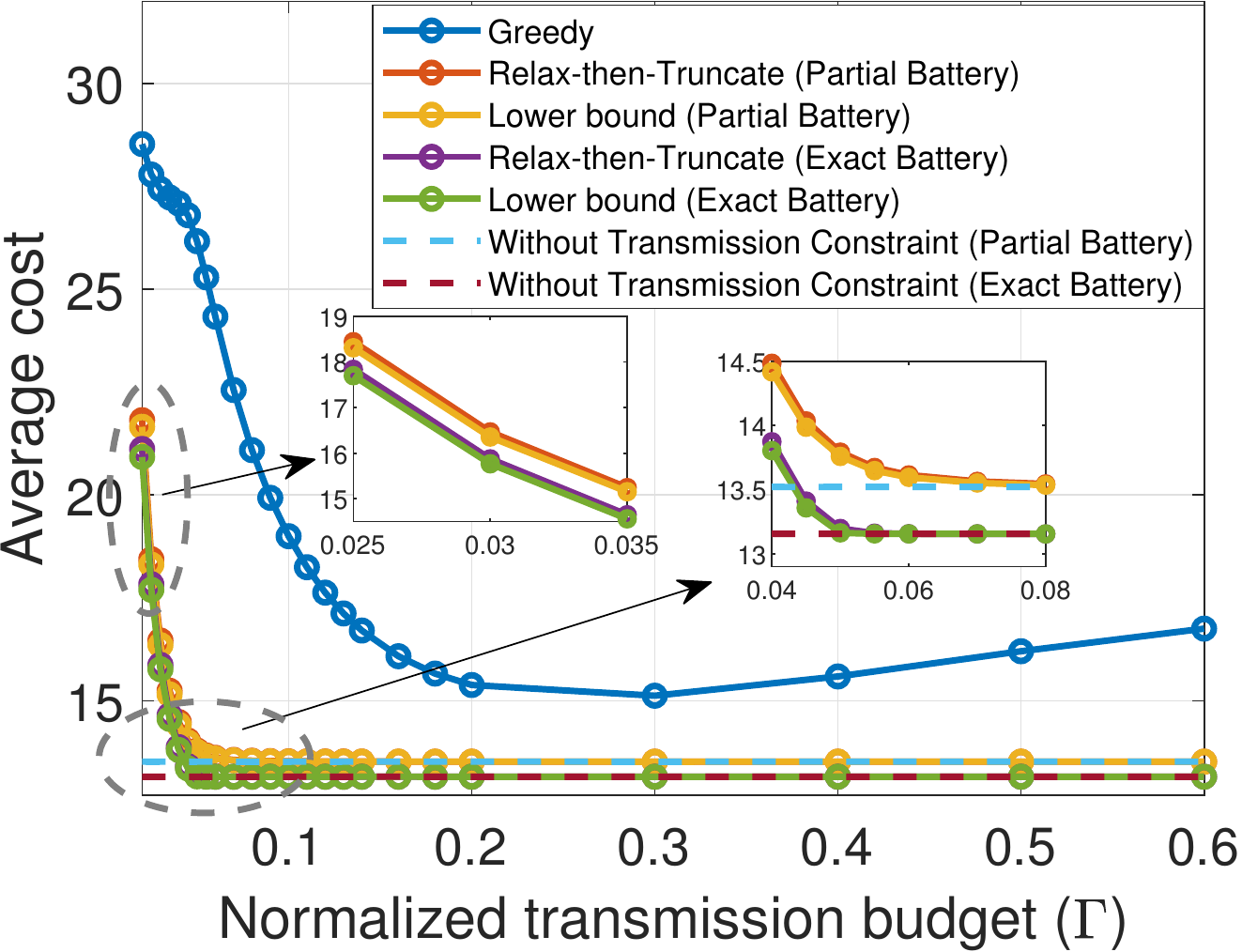}
}\vspace{-3mm}
\subfigure[]{
\includegraphics[width=\fw \columnwidth, trim={0 0 0 0},clip]{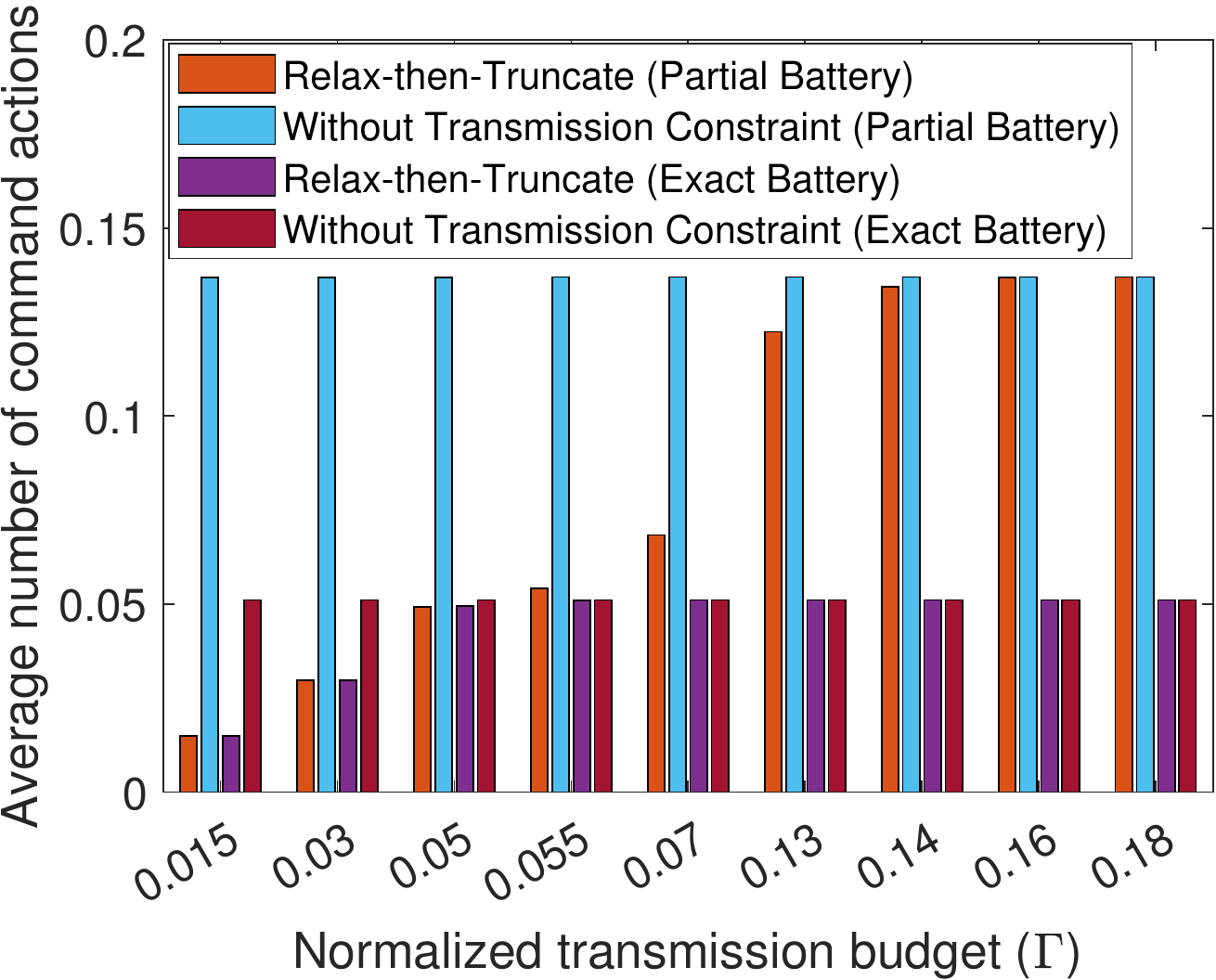}
}
\vspace{-3mm}
\caption{(a) Average cost and (b) Average number of command actions with respect to $\Gamma$ when ${K = 1000}$.}
\vspace{-5mm}
\label{fig_perf_K_fixed}
\end{figure}

$\text{Fig.\ \ref{fig_perf_K_fixed}(a)}$ and $\text{Fig.\ \ref{fig_perf_K_fixed}(b)}$
illustrate the average cost and the average number of command actions, respectively, with respect to the normalized transmission budget $\Gamma$. 
For the benchmarking, we plot the performance of an optimal policy for the case with no transmission constraint (i.e., ${N = K}$) \cite{hatami2021spawc,hatami2022spawc_partialbattery}.
As shown in $\text{Fig.\ \ref{fig_perf_K_fixed}(a)}$, the average cost for the proposed algorithm decreases as $\Gamma$ increases. This is because, for fixed $K$, the transmission budget $N$ increases by increasing $\Gamma$, and thus, the edge node can command more sensors at each slot to serve the users with fresh statuses more often.
Interestingly, there is a point after which increasing $\Gamma$ does not decrease the average cost. 
This is because, as shown in $\text{Fig.\ \ref{fig_perf_K_fixed}(b)}$, the average number of command actions stops increasing {(after ${\Gamma \geq 0.055}$ and ${\Gamma \geq 0.16}$ for the exact and partial battery knowledge, respectively)}, i.e., the constraint \eqref{st_opt2} becomes inactive, meaning that the edge node has more transmission budget than needed.
In these cases, the limited availability of energy at the EH sensors becomes a dominant factor in restraining the transmission of fresh status updates. 





\section{Conclusion}\label{sec_conclusions}
{We studied
on-demand AoI minimization in a multi-sensor EH IoT network where the status updating procedure leads to partial knowledge about the sensors' battery levels at the edge node.} 
We developed a low-complexity relax-then-truncate algorithm and proved that it is asymptotically optimal as the number of sensors goes to infinity. 
Numerical results showed that the relax-then-truncate algorithm reduces the average on-demand AoI roughly $30~\%$ compared to a request-aware greedy policy and that it has near-optimal performance even for moderate numbers of sensors, {which is} important for emerging IoT networks with hundreds of sensors connected.

\bibliographystyle{IEEEtran}
\bibliography{Bib/conf_short,Bib/IEEEabrv,Bib/Bibliography}

\end{document}